\documentclass[12pt]{article}

\usepackage[colorlinks, linkcolor=blue, citecolor=blue]{hyperref}           
\usepackage{color}
\usepackage{graphicx,subfigure,amsmath,amssymb,amsfonts,bm,epsfig,epsf,url,dsfont}
\usepackage{amsthm,mathrsfs}
\usepackage{bigints}

\newtheorem{thm}{Theorem}[section]
\newtheorem{lem}[thm]{Lemma}

\newtheorem{remark}[thm]{Remark}

\newtheorem{definition}[thm]{Definition}

\newtheorem{corollary}[thm]{Corollary}

\usepackage{tikz}
\usepackage{bbm}
\usepackage[small,bf]{caption}
\usepackage[top=1in,bottom=1in,left=1in,right=1in]{geometry}
\usepackage{fancybox}
\usepackage{hyperref}
\usepackage{algorithm}
\usepackage{algpseudocode}
\usepackage{verbatim}
\usepackage[titletoc,toc]{appendix}
\usepackage{amsmath}
\usepackage{array}
\usepackage{tabularx}
\usepackage{booktabs}
\usepackage{authblk}

\usepackage{mathtools}

\usepackage{scalerel,stackengine}
\stackMath
\newcommand\reallywidehat[1]{%
\savestack{\tmpbox}{\stretchto{%
  \scaleto{%
    \scalerel*[\widthof{\ensuremath{#1}}]{\kern-.6pt\bigwedge\kern-.6pt}%
    {\rule[-\textheight/2]{1ex}{\textheight}}%WIDTH-LIMITED BIG WEDGE
  }{\textheight}% 
}{0.5ex}}%
\stackon[1pt]{#1}{\tmpbox}%
}

\newcommand*{\rom}[1]{\expandafter\@slowromancap\romannumeral #1@}
\newcommand{\SO}{\mathrm{SO}}

\usepackage{xspace}

\numberwithin{equation}{section}

\allowdisplaybreaks

\begin{document}
\title{Moment Recovery from Tomographic Projections}
\title{Group-invariant moments under tomographic projections}

\author[1]{Amnon Balanov\thanks{Corresponding author: \url{amnonba15@gmail.com}}}
\author[1]{Tamir Bendory}
\author[2]{Dan Edidin}

\affil[1]{\normalsize School of Electrical and Computer Engineering, Tel Aviv University, Tel Aviv 69978, Israel}

\affil[2]{Department of Mathematics, University of Missouri, Columbia, MO 65211, USA}

\maketitle
\begin{abstract}
Let $f:\mathbb{R}^n\to\mathbb{R}$ be an unknown object, and suppose the observations are 
tomographic projections of randomly rotated copies of $f$ of the form $Y = P(R\cdot f)$, where $R$ is Haar-uniform in $\mathrm{SO}(n)$ and $P$ is the projection onto an $m$-dimensional subspace, so that $Y:\mathbb{R}^m\to\mathbb{R}$. 
We prove that, whenever $d\le m$, the $d$-th order moment of the projected data determines the full $d$-th order Haar-orbit moment of $f$, independently of the ambient dimension $n$.
We further provide an explicit algorithmic procedure for recovering the latter from the former. As a consequence, any identifiability result for the unprojected model based on the $d$-th order group-invariant moment extends directly to the tomographic setting at the same moment order. 
In particular, for $n=3$, $m=2$, and $d=2$, our result recovers a classical result in the cryo-EM literature: the covariance of the 2D projection images determines the second-order rotationally invariant moment of the underlying 3D object. 
\end{abstract}

\section{Introduction}

Tomographic imaging problems arise in a wide range of scientific settings in which an unknown object is observed only through lower-dimensional projections acquired at various orientations. In this work, the object is modeled as a function $f:\mathbb{R}^n\to\mathbb{R}$,  whereas each tomographic observation $Y$ is a function on $\mathbb{R}^m$ with $m<n$. Specifically, we consider measurements of the form
\begin{align}
    Y = P(R\cdot f) + \xi,
    \label{eqn:projective-model}
\end{align}
where $R\in \mathrm{SO}(n)$ is an unknown rotation, $(R\cdot f)(x)=f(R^{-1}x)$ denotes the natural action of $\mathrm{SO}(n)$ on the object, $P$ is the tomographic projection operator onto a fixed $m$-dimensional subspace, and $\xi$ is additive noise. Thus, $Y:\mathbb{R}^m\to\mathbb{R}$ is a noisy tomographic projection of a randomly oriented copy of $f$~\cite{natterer2001mathematics,herman2009fundamentals}.

A canonical example of~\eqref{eqn:projective-model} is single-particle cryo-electron microscopy (cryo-EM), where the unknown object is a 3D molecular structure and the data consist of a large collection of noisy 2D projection images acquired at random viewing directions \cite{bendory2020single,singer2020computational}. More broadly, this paradigm belongs to the class of \emph{orbit-recovery} problems, in which an unknown signal is observed only after an action of a latent group element, possibly followed by corruption or partial observation \cite{bandeira2023estimation,fan2024maximum}.

A useful reference point is the corresponding unprojected orbit-recovery model, in which one observes
\begin{align}
    Y = R\cdot f + \xi,
    \label{eqn:unprojected-model}
\end{align}
without the projection operator $P$. In the high-noise regime, which is often the statistically most challenging regime of interest, a central theme in orbit recovery is that identifiability and sample complexity are governed by the lowest-order moment that distinguishes the orbit of the underlying signal. For many compact group actions in the unprojected setting, this cut-off moment is by now well understood, and it forms the basis of a broad literature on moment-based and invariant-based methods in signal processing, imaging, and structural biology \cite{bandeira2023estimation,bhamre2015orthogonal,levin20183d,bendory2017bispectrum,bendory2026provable}. By contrast, much less is understood in the tomographic setting~\eqref{eqn:projective-model}. In particular, it is not clear a priori whether tomographic projection preserves the same low-order orbit information, nor whether the minimal informative moment in the unprojected model remains sufficient after tomographic projection.

This question is already implicit in the classical cryo-EM literature. In cryo-EM, Kam's theorem shows that, under uniformly distributed viewing directions, the covariance of the 2D projection images determines the second-order rotationally invariant moment of the underlying 3D volume \cite{kam1980reconstruction,bendory2024sample,bhamre2015orthogonal}. In other words, although the volume itself is never directly observed, the full second-order rotationally invariant moment 
%encoded by its Fourier-domain moment tensor 
can still be recovered from the projected data. This naturally raises a broader question: when are the projected model~\eqref{eqn:projective-model} and the unprojected model~\eqref{eqn:unprojected-model} equivalent at the level of moments?

The purpose of this work is to answer this question. We study orbit recovery from tomographic projections under random rotations and ask when the $d$-th order moments of the projected data determine the corresponding $d$-th order %orbit
moments of the underlying object. Our main result shows that, under Haar-uniform distribution over $\SO(n)$, this equivalence holds whenever $d\le m$: the $d$-th order  moment of the projected data determines the full %Haar-orbit 
$d$-th order moment of the object, regardless of the ambient dimension $n$. Consequently, any orbit-identifiability result for the unprojected model that is formulated in terms of the $d$-th order Haar-orbit moment transfers immediately to the tomographic setting under the same conditions. 

In particular, Kam's classical cryo-EM relation appears as the special case $(n,m,d)=(3,2,2)$. 
We note, however, that the derivations in the cryo-EM literature, such as Kam's original analysis~\cite{kam1980reconstruction}, establish the second-order result in a highly explicit and model-specific manner.
Our result shows that this phenomenon is not specific to the particular case $(n,m,d)=(3,2,2)$, but follows from a 
a general statement in Fourier slice geometry: Whenever $d \le m$, any $d$ frequencies lie in a common $m$-dimensional slice and the Haar invariance of the rotational distribution implies that the projected $d$-th order moment determines the full $d$-th order moment.

The underlying intuition is illustrated in Figure~\ref{fig:1}. By the Fourier slice theorem, the Fourier transform of a tomographic projection is the restriction of the ambient Fourier transform to a rotated central slice; see Figure~\ref{fig:1}(a). In the cryo-EM setting, this implies that any two frequencies lie in a common 2D slice. More generally, if $d\le m$, then any $d$ frequencies lie in a common $m$-dimensional slice; see Figure~\ref{fig:1}(b)-(c). The proof turns this geometric observation into a formal theorem by combining the Fourier slice theorem with Haar invariance, which allows one to transport any such slice back to a reference projection geometry.
More broadly, the present work identifies the threshold $d\le m$ as the natural regime in which $d$-th order orbit information is preserved under tomographic projection. This provides a direct bridge between moment-based identifiability in the unprojected model and moment-based identifiability in the tomographic model.

\begin{figure*}[t!]
    \centering
    \includegraphics[width=0.95 \linewidth]{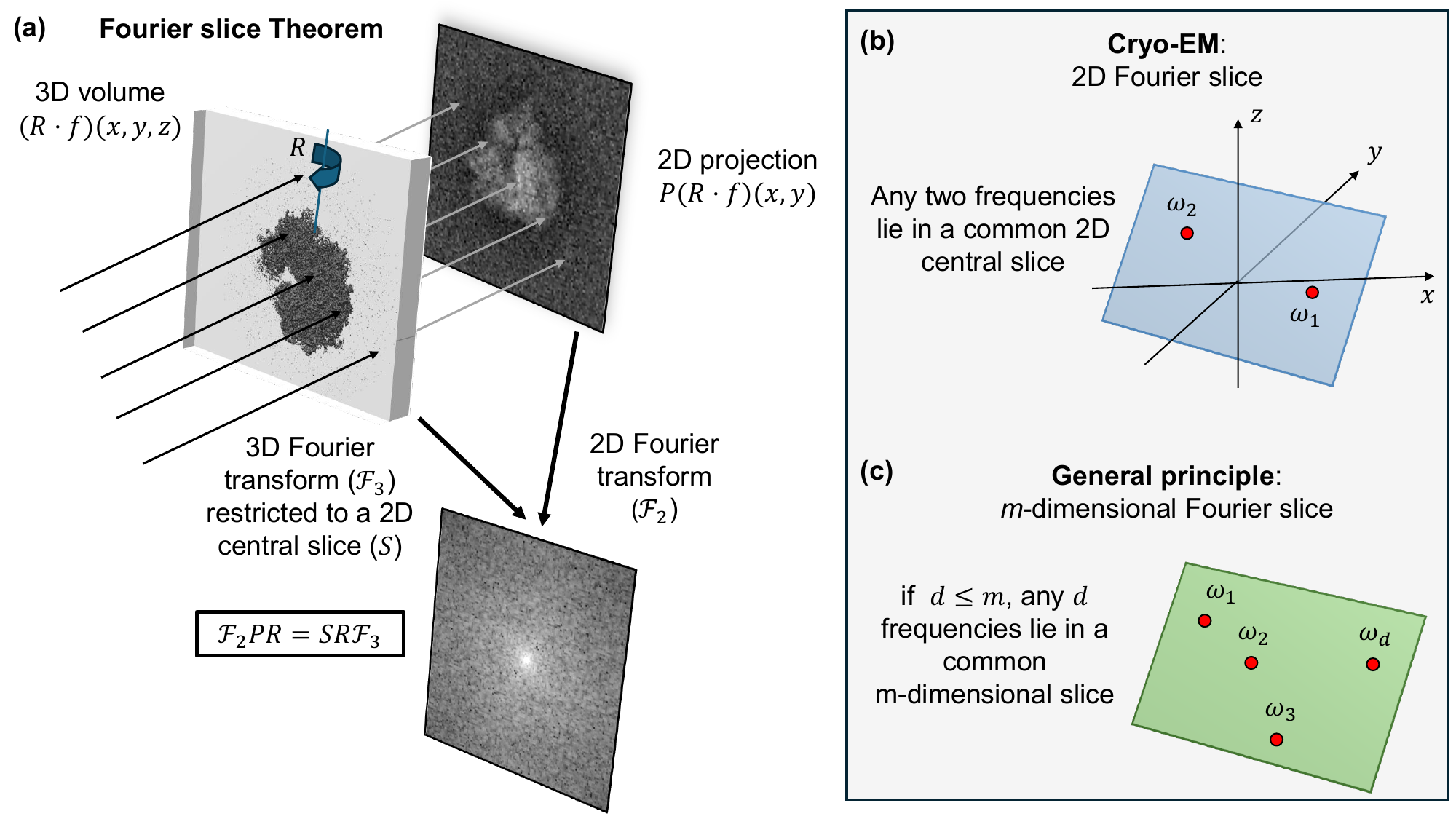}
    \caption{ \textbf{Illustration of the moment-lifting mechanism.} \textbf{(a)} Fourier slice theorem} in 3D: the Fourier transform of a tomographic projection $P(R\cdot f)$ equals the restriction of the ambient Fourier transform $\widehat f$ to a rotated 2D central slice. \textbf{(b)} Cryo-EM case $(n,m,d)=(3,2,2)$: any two frequencies lie in a common 2D slice, explaining why second-order projection statistics determine the full second-order rotational invariant. \textbf{(c)} General case: if $d\le m$, any $d$ frequencies lie in a common $m$-dimensional slice. This is the underlying geometric reason why projected $d$-th order moments determine the full %Haar-orbit
    $d$-th order moments.
    \label{fig:1} 
\end{figure*}

\section{Preliminaries}

\paragraph{Notation.}
We write $\mathbb{R}^n$ and $\mathbb{C}$ for the $n$-dimensional real Euclidean space and the complex numbers, respectively. The group of $n\times n$ rotation matrices is denoted by $\mathrm{SO}(n)$, and $\mathrm{Haar}(\mathrm{SO}(n))$ stands for the normalized Haar measure on $\mathrm{SO}(n)$. For $x,y \in \mathbb{R}^n$, $\langle x,y\rangle$ denotes the standard Euclidean inner product. Throughout, we write $\mathcal{F}_n$ and $\mathcal{F}_m$ for the Fourier transforms on $\mathbb{R}^n$ and $\mathbb{R}^m$, respectively. Explicitly, for $f\in L^1(\mathbb{R}^n)$,
\begin{align}
    \widehat f(\omega) = (\mathcal{F}_n f)(\omega) \triangleq \int_{\mathbb{R}^n} f(x)\, e^{-i\langle \omega,x\rangle}\, dx,
    \qquad \omega\in\mathbb{R}^n.
\end{align}
For a subspace $E \subset \mathbb{R}^n$, $S_E$ denotes restriction of a function on $\mathbb{R}^n$ to $E$. Throughout, expectations are understood whenever the relevant integrability conditions hold.

\subsection{Tomographic observation model}

Recall that we consider the tomographic observation model in which the observations $Y:\mathbb{R}^m\to\mathbb{R}$ are given by
\begin{align}
    Y = P(R \cdot f) + \xi,
    \label{eqn:projective-model2}
\end{align}
where $f:\mathbb{R}^n\to\mathbb{R}$ is the unknown object, $R \in \mathrm{SO}(n)$ is an unknown random rotation acting on $f$ via $(R\cdot f)(x)\triangleq f(R^{-1}x)$ and $\xi$ is an additive Gaussian noise with a variance level $\sigma^2$. We assume that the rotations $R$ are drawn independently from the Haar-uniform distribution on $\mathrm{SO}(n)$.
%\paragraph{Population and finite-sample quantities.}
Since our focus is on population moments, and since under independent noise with known distribution the noise contribution to these moments is explicit, we henceforth omit the noise term and work with the noiseless model $Y=P(R\cdot f)$.
In particular, for Gaussian noise, the corresponding signal-only projected moments can be estimated from noisy observations by using appropriate Hermite/Wick-debiased polynomials; see, e.g., \cite[Section~4.1 and Appendix~C]{balanov2026generalized}.
The operator $P$ is a tomographic projection onto a fixed $m$-dimensional subspace. More precisely, writing $(u,v)\in\mathbb{R}^m\times\mathbb{R}^{n-m}$, we define
\begin{align}
    P(f)(u) \triangleq \int_{\mathbb{R}^{n-m}} f(u,v)\, dv,
    \qquad u\in\mathbb{R}^m.
\end{align}
Thus, $P(f)$ is obtained by integrating $f$ along the last $n-m$ coordinates.

\subsection{The Fourier slice theorem}
%\subsection{Projection-slice and Fourier-slice theorem}

A basic structural fact underlying tomography is that projection in the spatial domain corresponds to restriction to a lower-dimensional central slice in the Fourier domain. In addition, it is well-known that the Fourier operator and rotations commute; thus, after rotation of the object, the corresponding Fourier slice rotates accordingly. We record this standard fact in the following theorem using the notation of the present work, and include a short proof for completeness.

To formulate this identity, it is convenient to identify the observation space $\mathbb{R}^m$ with a fixed reference subspace of $\mathbb{R}^n$. We therefore introduce the canonical embedding $\iota:\mathbb{R}^m\to\mathbb{R}^n$,
\begin{align}
    \iota(\eta) \triangleq (\eta,0),
    \qquad \eta\in\mathbb{R}^m.
    \label{eqn:canonical-embedding}
\end{align}
In particular, $\iota(\mathbb{R}^m)=\mathbb{R}^m\times\{0\}\subset\mathbb{R}^n$  is the reference $m$-dimensional subspace, and the slices corresponding to different viewing directions are precisely its rotated copies. 
%\tamir{we should be very clear that this is not our result, but a classical result, and put it here with a short proof for completeness and to fix notation} \amnon{Added above.}

\begin{thm}[Fourier slice theorem]
\label{thm:fourier_slice_rot} 
Let $f \in L^1(\mathbb{R}^n)$ and $R \in \mathrm{SO}(n)$. Then, for every $\eta \in \mathbb{R}^m$,
\begin{align}
    \widehat{P(R\cdot f)}(\eta) &= \widehat f(R^{-1}\iota(\eta)). \label{eq:fourier_slice_theorem_rot}
\end{align}
\end{thm}

\begin{proof}[Proof of Theorem~\ref{thm:fourier_slice_rot}]
By definition of $P$ and of the $m$-dimensional Fourier transform,
\begin{align}
    \widehat{P(R\cdot f)}(\eta) = \int_{\mathbb{R}^m} \left( \int_{\mathbb{R}^{n-m}} f(R^{-1}(u,v))\,dv \right) e^{-i\langle \eta,u\rangle}\,du .
\end{align}
Identifying $(u,v)\in\mathbb{R}^m\times\mathbb{R}^{n-m}$ with $x\in\mathbb{R}^n$, and using $\iota(\eta)=(\eta,0)$, this becomes
\begin{align}
    \widehat{P(R\cdot f)}(\eta) = \int_{\mathbb{R}^n} f(R^{-1}x)e^{-i\langle \iota(\eta),x\rangle}\,dx .
\end{align}
Changing variables $x=Rz$, using $|\det R|=1$, and using the orthogonality identity $R^\top=R^{-1}$, we obtain
\begin{align}
    \widehat{P(R\cdot f)}(\eta) = \int_{\mathbb{R}^n} f(z)e^{-i\langle \iota(\eta),Rz\rangle}\,dz = \int_{\mathbb{R}^n} f(z)e^{-i\langle R^{-1}\iota(\eta),z\rangle}\,dz = \widehat f(R^{-1}\iota(\eta)).
\end{align}
\end{proof}

Theorem~\ref{thm:fourier_slice_rot} shows that each tomographic projection reveals the values of the Fourier transform of the object on an $m$-dimensional central slice. After rotation of the object, the observed Fourier data correspond to the restriction of $\widehat f$ to the rotated slice $R^{-1}\iota(\mathbb{R}^m)$. This fundamental geometric observation is the basis for the moment-lifting result established below.

\subsection{Fourier-domain moment tensors}

Since the tomographic structure of the problem is most naturally expressed in the Fourier domain through the Fourier slice theorem, we work throughout with Fourier-domain moment tensors. This entails no loss of information, since Fourier-domain moments are equivalent to their real-space counterparts by Fourier inversion. Because Fourier transforms are generally complex-valued, the associated moment tensors are naturally complex-valued as well, even when the underlying object is real-valued.

\begin{definition}[Full and projected Fourier-domain moment tensors]
\label{def:fourier_moment_tensors}
Let $f:\mathbb{R}^n\to\mathbb{R}$, let $R\sim \mathrm{Haar}(\mathrm{SO}(n))$, and let $P$ be the tomographic projection operator onto the fixed reference $m$-dimensional subspace. For $d\ge 1$, define the \emph{full %Haar-orbit
$d$-th order Fourier moment tensor}, associated with~\eqref{eqn:unprojected-model}, by
\begin{align}
    M_{f, \mathrm{full}}^{(d)}(\omega_1,\dots,\omega_d) \triangleq \mathbb{E}_{R\sim\mathrm{Haar}(\mathrm{SO}(n))} \left[\prod_{j=1}^d \widehat f(R^{-1}\omega_j)\right],
    \qquad    \omega_1,\dots,\omega_d\in\mathbb{R}^n,
\end{align}
and the \emph{projected $d$-th order Fourier moment tensor}, associated with~\eqref{eqn:projective-model}, by
\begin{align}
    M_{f, \mathrm{proj}}^{(d)}(\eta_1,\dots,\eta_d) \triangleq \mathbb{E}_{R\sim\mathrm{Haar}(\mathrm{SO}(n))} \left[\prod_{j=1}^d \widehat{P(R\cdot f)}(\eta_j)\right],
    \qquad    \eta_1,\dots,\eta_d\in\mathbb{R}^m.
\end{align}
\end{definition}
The main question is under what conditions the projected moment tensor $M_{f,\mathrm{proj}}^{(d)}$ determines the full %Haar-orbit 
moment tensor $M_{f,\mathrm{full}}^{(d)}$.

\section{Main results}
In this section we present the main theoretical result of the paper and its consequences for tomographic orbit recovery. We then describe the corresponding constructive procedure (i.e., an algorithm) for evaluating the full %Haar-orbit 
moment from the projected moment.

\subsection{Moment equivalence under tomographic projection}

We now state the main result of this work. 
%Its significance is both geometric and statistical. 
Geometrically, it shows that under Haar-uniform rotations, the $d$-th order Fourier moment of the tomographic projections contains the same information as the full Haar-orbit $d$-th order Fourier moment of the underlying volume, provided that $d \le m$. %Statistically, this implies
This, in turn, implies that any recovery principle based on the full $d$-th order orbit moment transfers directly to the tomographic setting without increasing the required moment order. In this sense, the theorem below establishes a general principle showing that tomographic projected moments determine the corresponding full 
%Haar-orbit 
moments.

\begin{thm}[Recovery of full %Haar-orbit 
$d$-th order moments from tomographic projected moments]
\label{thm:moment_lifting_tomography}
Assume that the viewing directions are Haar-uniform on $\mathrm{SO}(n)$. Let $\iota$ denote the canonical embedding defined in~\eqref{eqn:canonical-embedding}, and let $M_{f, \mathrm{full}}^{(d)}$ and $M_{f, \mathrm{proj}}^{(d)}$ be the full and projected Fourier-domain moment tensors introduced in Definition~\ref{def:fourier_moment_tensors}. Then,  the following statements hold.

\begin{enumerate}
    \item For every $\eta_1,\dots,\eta_d \in \mathbb{R}^m$,
    \begin{align}
        M_{f, \mathrm{proj}}^{(d)}(\eta_1,\dots,\eta_d) = M_{f, \mathrm{full}}^{(d)}(\iota(\eta_1),\dots,\iota(\eta_d)).        \label{eq:proj_full_restriction}
    \end{align}

    \item Let $\omega_1,\dots,\omega_d \in \mathbb{R}^n$, and suppose that there exists an $m$-dimensional subspace $E \subset \mathbb{R}^n$ such that $\omega_1,\dots,\omega_d \in E$. Then, there exist $Q \in \mathrm{SO}(n)$ and $\eta_1,\dots,\eta_d \in \mathbb{R}^m$ such that
    \begin{align}
        \omega_j = Q \iota( \eta_j),
        \qquad j=1,\dots,d,
        \label{eqn:related-by-rotation}
    \end{align}
    and
    \begin{align}
        M_{f, \mathrm{full}}^{(d)}(\omega_1,\dots,\omega_d) = M_{f, \mathrm{proj}}^{(d)}(\eta_1,\dots,\eta_d).    \label{eq:full_from_proj_general}
    \end{align}

    \item If $d \le m$, then $M_{f, \mathrm{proj}}^{(d)}$ determines $M_{f, \mathrm{full}}^{(d)}$ on all of $(\mathbb{R}^n)^d$.
\end{enumerate}
\end{thm}

\begin{proof}[Proof of Theorem~\ref{thm:moment_lifting_tomography}]
We first prove~\eqref{eq:proj_full_restriction}. By the Fourier slice theorem (Theorem~\ref{thm:fourier_slice_rot}), for every $\eta \in \mathbb{R}^m$,
\begin{align}
    \widehat{P(R\cdot f)}(\eta) =
    \widehat{f}(R^{-1}\iota(\eta)).
\end{align}
Substituting this identity into the definition of $M_{f, \mathrm{proj}}^{(d)}$, we obtain
\begin{align}
    M_{f, \mathrm{proj}}^{(d)}(\eta_1,\dots,\eta_d) &= \mathbb{E}_{R} \left[\prod_{j=1}^d \widehat{P(R\cdot f)}(\eta_j)\right] \\
    &= \mathbb{E}_{R} \left[\prod_{j=1}^d \widehat{f}(R^{-1}\iota(\eta_j)) \right] \\
    &= M_{f, \mathrm{full}}^{(d)}(\iota(\eta_1),\dots,\iota(\eta_d)),
\end{align}
which proves~\eqref{eq:proj_full_restriction}.

Next, suppose that $\omega_1,\dots,\omega_d$ lie in a common $m$-dimensional subspace $E$. Since every $m$-dimensional subspace is a rotation of the reference subspace $\iota(\mathbb{R}^m)$, there exists $Q \in \mathrm{SO}(n)$ such that
\begin{align}
    E = Q \iota(\mathbb{R}^m).
\end{align}
Hence, for each $j$, there exists $\eta_j \in \mathbb{R}^m$ such that
\begin{align}
    \omega_j = Q \iota( \eta_j).
\end{align}
Then
\begin{align}
    M_{f, \mathrm{full}}^{(d)}(\omega_1,\dots,\omega_d) &=    \mathbb{E}_{R}\left[\prod_{j=1}^d \widehat{f}(R^{-1}Q \iota(\eta_j))
    \right].
\end{align}
Using the left-invariance of Haar measure on $\mathrm{SO}(n)$, the random variable $\widetilde{R} = Q^{-1}R$ is again Haar-distributed. Therefore,
\begin{align}
    M_{f, \mathrm{full}}^{(d)}(\omega_1,\dots,\omega_d) &= \mathbb{E}_{\widetilde{R}} \left[\prod_{j=1}^d \widehat{f}(\widetilde{R}^{-1}\iota(\eta_j)) \right] \\
    &= M_{f, \mathrm{full}}^{(d)}(\iota(\eta_1),\dots,\iota(\eta_d)).
\end{align}
Combining this with~\eqref{eq:proj_full_restriction} yields
\begin{align}
    M_{f, \mathrm{full}}^{(d)}(\omega_1,\dots,\omega_d) = M_{f, \mathrm{proj}}^{(d)}(\eta_1,\dots,\eta_d),
\end{align}
which proves~\eqref{eq:full_from_proj_general}.

Finally, if $d \le m$, then any $d$ vectors $\omega_1,\dots,\omega_d \in \mathbb{R}^n$ span a subspace of dimension at most $d$, and hence at most $m$. Therefore, they lie in some $m$-dimensional subspace $E$, and the previous part shows that $M_{f, \mathrm{full}}^{(d)}(\omega_1,\dots,\omega_d)$ is determined by $M_{f, \mathrm{proj}}^{(d)}$. Since this holds for every $(\omega_1,\dots,\omega_d) \in (\mathbb{R}^n)^d$, the full $d$-th order Fourier moment $M_{f, \mathrm{full}}^{(d)}$ is completely determined by $M_{f, \mathrm{proj}}^{(d)}$.
\end{proof}
Theorem~\ref{thm:moment_lifting_tomography} identifies the threshold $d \le m$ as the precise regime in which the $d$-th order orbit moment remains fully visible through $m$-dimensional tomographic projections. Indeed, any $d$ frequencies in $\mathbb{R}^n$ span a subspace of dimension at most $d$, and therefore, when $d \le m$, they can all be embedded into a common $m$-dimensional slice. Haar invariance then allows one to transfer this slice back to the fixed reference projection geometry. In this way, the theorem shows that tomographic projection preserves $d$-th order orbit information whenever the projection dimension is at least the moment order.

An immediate consequence is that any identifiability result formulated in terms of the full %Haar-orbit 
$d$-th order moment automatically yields a corresponding identifiability result in the tomographic model. This is the content of the following corollary.

\begin{corollary}[Orbit recovery from projected $d$-th order moments]
\label{cor:orbit_recovery_from_proj_moments}
Assume the setting of Theorem~\ref{thm:moment_lifting_tomography}, and suppose that the %Haar-orbit 
$d$-th order Fourier moment $M_{f, \mathrm{full}}^{(d)}$ uniquely determines the orbit of $f$ under the action of $\mathrm{SO}(n)$. If $d \le m$, then the projected $d$-th order Fourier moment $M_{f, \mathrm{proj}}^{(d)}$ also uniquely determines the orbit of $f$.
\end{corollary}

\begin{proof}[Proof of Corollary~\ref{cor:orbit_recovery_from_proj_moments}]
By Theorem~\ref{thm:moment_lifting_tomography}, if $d \le m$, then $M_{f, \mathrm{proj}}^{(d)}$ determines $M_{f, \mathrm{full}}^{(d)}$. By assumption, $M_{f, \mathrm{full}}^{(d)}$ uniquely determines the orbit of $f$. Hence $M_{f, \mathrm{proj}}^{(d)}$ uniquely determines the orbit of $f$ as well.
\end{proof}

Corollary~\ref{cor:orbit_recovery_from_proj_moments} captures the main conceptual implication of our result. While a substantial literature has developed moment-identifiability results for orbit recovery in the unprojected model, comparatively little is known in the tomographic setting. Our theorem and corollary uncover a geometric bridge between these two regimes at the level of moments: under Haar-uniform viewing directions, and whenever $d\le m$, the $d$-th order moment of the tomographic data determines the same $d$-th order %Haar-orbit 
moment that appears in the unprojected model. As a result, any identifiability statement for the unprojected model formulated in terms of $d$-th order orbit moments immediately carries over to the tomographic setting. In particular, if the minimal informative moment order in the unprojected model is $d$, then the same order remains sufficient after projection.

As a direct corollary of Theorem~\ref{thm:moment_lifting_tomography}, one recovers the classical cryo-EM second-order moment phenomenon: the second-order statistics of random two-dimensional projections determine the full second-order rotational invariant of the underlying three-dimensional volume. This property is used, for example, to estimate the covariance of a 3D molecular structure from its tomographic projections~\cite{anden2018structural,gilles2025cryo}.

\begin{corollary}[Cryo-EM second-order moment recovery]
\label{cor:cryoEM-second-moment}
Consider the cryo-EM setting $n=3$ and $m=2$. Then, the projected second-order Fourier moment tensor determines the full %Haar-orbit
second-order Fourier moment tensor. Equivalently, for every pair of frequencies $\omega_1,\omega_2 \in \mathbb{R}^3$, the value of $M^{(2)}_{f,\mathrm{full}}(\omega_1,\omega_2)$ can be recovered from $M^{(2)}_{f,\mathrm{proj}}$.
\end{corollary}

\begin{remark}[On the role of Haar-uniform viewing directions]
The proof of Theorem~\ref{thm:moment_lifting_tomography} uses Haar invariance in an essential way. Indeed, if $R\sim \mathrm{Haar}(\mathrm{SO}(n))$, then for every fixed $Q\in \mathrm{SO}(n)$, the random matrices $R$ and $Q^{-1}R$ have the same distribution. This is precisely what makes all $m$-dimensional central slices statistically equivalent under rotations, and therefore allows the projected moment to depend only on the queried frequencies, rather than on the particular slice used to represent them.

For a general non-uniform viewing distribution, this rotational equivalence is lost: different slices are sampled with different weights, so the projected $d$-th order moment may depend on how the containing slice is embedded in $\mathbb{R}^n$. Consequently, the projected moment need not coincide directly with a slice-independent Haar-orbit moment of the underlying object.

If, however, the viewing distribution is known and admits a strictly positive density with respect to Haar measure, then one may in principle compensate for the non-uniform sampling by reweighting the observations, thereby converting expectations under the given distribution into Haar expectations. In that case, an analogue of the theorem may still hold after such normalization. By contrast, when the viewing distribution is unknown, or when its density vanishes on a set of positive Haar measure, such a reduction is generally unavailable.
\end{remark}

\subsection{Constructive recovery of the full moment tensor}

We now describe a constructive procedure, derived from the proof of Theorem~\ref{thm:moment_lifting_tomography}, for evaluating the full 
%Haar-orbit 
$d$-th order Fourier moment tensor from its tomographic counterpart. Assume throughout that the viewing directions are Haar-uniform on $\mathrm{SO}(n)$, and let $d\le m$. Given a query tuple $(\omega_1,\dots,\omega_d)\in(\mathbb{R}^n)^d$, our goal is to evaluate $M_{f,\mathrm{full}}^{(d)}(\omega_1,\dots,\omega_d)$ using only the projected moment tensor $M_{f,\mathrm{proj}}^{(d)}$.

Let $\widetilde E\subset \mathbb{R}^n$ be any $m$-dimensional linear subspace containing $\omega_1,\ldots,\omega_d$, and let $Q\in \mathrm{SO}(n)$ satisfy $\widetilde E = Q\,\iota(\mathbb{R}^m)$.
Then, there exist unique vectors $\eta_1,\ldots,\eta_d\in \mathbb{R}^m$ such that $\omega_j = Q\,\iota(\eta_j)$ for each $j = 1, \ldots d$.

\begin{lem}
\label{lem:projected-to-unprojected}
With the notation above, there is an equality of moments,
\begin{align}
    M_{f,\mathrm{full}}^{(d)}(\omega_1,\dots,\omega_d)
    =    M_{f,\mathrm{proj}}^{(d)}(\eta_1,\dots,\eta_d).
\end{align}
\end{lem}

\begin{proof} By part~(2) of Theorem~\ref{thm:moment_lifting_tomography}, whenever $\omega_1,\dots,\omega_d$ lie in a common $m$-dimensional subspace and admit a representation of the form $\omega_j = Q\,\iota(\eta_j)$, for every $j=1,\dots,d$, 
one has $M_{f,\mathrm{full}}^{(d)}(\omega_1,\dots,\omega_d) = M_{f,\mathrm{proj}}^{(d)}(\eta_1,\dots,\eta_d)$.
This proves the claim.

\end{proof}

Lemma~\ref{lem:projected-to-unprojected} immediately yields a constructive evaluation rule: to compute the full moment at an arbitrary query tuple, it suffices to choose any $m$-dimensional slice containing the queried frequencies, express those frequencies in the coordinates of that slice, and evaluate the projected moment at the resulting coordinates. We summarize the procedure in Algorithm~\ref{alg:momentRecovery}.

Let $E \coloneqq \operatorname{span}\{\omega_1,\dots,\omega_d\}$ and $r \coloneqq \dim(E)$, so that $r\le d\le m$. Choose an orthonormal basis $u_1,\dots,u_r$ of $E$, for example by applying Gram-Schmidt to $\omega_1,\dots,\omega_d$ and discarding any zero vectors that arise. Since $r\le m$, extend this basis to an orthonormal family $u_1,\dots,u_m$ spanning an $m$-dimensional subspace $\widetilde E\subset\mathbb{R}^n$ that contains $E$, and then extend further to an orthonormal basis $u_1,\dots,u_n$ of $\mathbb{R}^n$. Let $Q\in O(n)$ be the orthogonal matrix whose columns are $u_1,\dots,u_n$. By construction, the first $m$ columns of $Q$ are exactly $u_1,\dots,u_m$, and therefore
$Q\iota(\mathbb{R}^m)=\operatorname{span}\{u_1,\dots,u_m\}=\widetilde E$, where $\iota:\mathbb{R}^m\to\mathbb{R}^n$ is the canonical embedding introduced in \eqref{eqn:canonical-embedding}.

Since each $\omega_j$ lies in $\widetilde E$, the vector $Q^\top \omega_j$ has vanishing last $n-m$ coordinates. Thus, $Q^\top\omega_j\in \iota(\mathbb{R}^m)$, and hence there exists a unique vector $\eta_j\in\mathbb{R}^m$ such that
\begin{align}
    Q^\top\omega_j=\iota(\eta_j),    \qquad\text{equivalently,}\qquad
    \omega_j=Q\iota(\eta_j).
\end{align}

A priori, $Q$ need not belong to $\mathrm{SO}(n)$. However, when $n>m$, this causes no difficulty. Indeed, if $Q\notin \mathrm{SO}(n)$, replace any column $u_k$ with $k>m$ by $-u_k$, and denote the resulting matrix by $Q'$. Then, $Q'\in \mathrm{SO}(n)$, while
$(Q')^\top\omega_j = Q^\top\omega_j$ for every $j$, because $\omega_j\in \widetilde E=\operatorname{span}\{u_1,\dots,u_m\}$ is orthogonal to every $u_k$ with $k>m$. In particular,
\begin{align}
    \omega_j = Q'\iota(\eta_j)
    \qquad\text{for all } j=1,\dots,d.
\end{align}
Therefore, by Lemma~\ref{lem:projected-to-unprojected},
$M_{f,\mathrm{full}}^{(d)}(\omega_1,\dots,\omega_d) = M_{f,\mathrm{proj}}^{(d)}(\eta_1,\dots,\eta_d)$.
Accordingly, evaluating the full moment at an arbitrary query tuple reduces to expressing the query frequencies in coordinates of any $m$-dimensional slice containing them and then evaluating the projected moment at the resulting coordinates. 

\begin{algorithm}[t!]
  \caption{Recovery of the full moment from the projected moment}
  \label{alg:momentRecovery}
\textbf{Input:} A tuple $(\omega_1,\ldots,\omega_d)\in(\mathbb{R}^n)^d$ with $d\le m$, and the projected moment tensor $M_{f,\mathrm{proj}}^{(d)}$.\\
\textbf{Output:} The value of $M_{f,\mathrm{full}}^{(d)}(\omega_1,\ldots,\omega_d)$. \\
\textbf{Procedure:}
\begin{enumerate}
  \item Let $E=\operatorname{span}\{\omega_1,\ldots,\omega_d\}$. Since $\dim(E)\le d\le m$, construct an orthonormal family $u_1,\ldots,u_m$ spanning an $m$-dimensional subspace $\widetilde E\subseteq\mathbb{R}^n$ containing $E$, for example by applying Gram-Schmidt to $\omega_1,\ldots,\omega_d$ and extending the resulting family.

    \item Extend $u_1,\ldots,u_m$ to an orthonormal basis $u_1,\ldots,u_n$ of $\mathbb{R}^n$, and let $Q\in O(n)$ be the matrix with columns $u_1,\ldots,u_n$.  Then, $\widetilde E=Q\iota(\mathbb{R}^m)$.
    
  \item For each $j=1,\ldots,d$, let $\eta_j\in\mathbb{R}^m$ be the vector of the first $m$ coordinates of $Q^\top\omega_j$, equivalently, the unique vector such that $\omega_j=Q\iota(\eta_j)$.

  \item Return
  \[
      M_{f,\mathrm{full}}^{(d)}(\omega_1,\ldots,\omega_d) = M_{f,\mathrm{proj}}^{(d)}(\eta_1,\ldots,\eta_d).
  \]
\end{enumerate}
\end{algorithm}

\paragraph{Computational cost.}
Algorithm~\ref{alg:momentRecovery}  extends an orthonormal basis of $E=\operatorname{span}\{\omega_1,\ldots,\omega_d\}$ to a full orthonormal basis of $\mathbb{R}^n$ in order to produce an orthogonal matrix $Q$.
A naive implementation of this full completion may require $O(n^3)$ operations.
However, the full matrix $Q$ is not needed computationally.
Let $r=\dim(E)\le d$, and let $U_r\in\mathbb{R}^{n\times r}$ be an orthonormal basis matrix for $E$, obtained for example by applying Gram--Schmidt or QR decomposition to the $n\times d$ matrix with columns $\omega_1,\ldots,\omega_d$.
Then, each $\omega_j$ is represented by the coordinate vector $U_r^\top\omega_j\in\mathbb{R}^r$, and one may define
\begin{align}
    \eta_j = \bigl(U_r^\top\omega_j,0,\ldots,0\bigr) \in \mathbb{R}^m.
\end{align}
Thus, for a single query tuple $(\omega_1,\ldots,\omega_d)$, the full $n\times n$ matrix $Q$ need not be formed, and the geometric lifting step can be implemented in $O(nd^2)$ operations.
In particular, for fixed moment order $d$, this cost is linear in the ambient dimension $n$.

\section{Discussion and outlook}

A key conceptual implication of Theorem~\ref{thm:moment_lifting_tomography} is that the threshold $d \le m$ arises directly from the geometry of Fourier slices. Indeed, any $d$ frequencies span a subspace of dimension at most $d$, and hence can be contained in a common $m$-dimensional central slice precisely when $d \le m$. Haar invariance then allows one to transport information from an arbitrary such slice back to the fixed reference projection geometry. In this sense, the condition $d \le m$ is not merely technical, but reflects the intrinsic geometry of the tomographic model.

At the same time, our results also delineate the scope of this mechanism. First, the argument relies essentially on Haar-uniform viewing directions. The invariance of Haar measure is used in a crucial way, and the conclusion need not persist under non-uniform orientation distributions. This issue is especially relevant in cryo-EM, where viewing directions are often markedly non-uniform~\cite{baldwin2020non, naydenova2017measuring}. In that setting, the projected $d$-th order moment is no longer simply related to a rotationally invariant full moment, and additional ambiguities or degeneracies may arise. 
Extending the present moment-equivalence principle to broader families of distributions over $\SO(n)$ remains an important open problem.

Second, while the theorem shows that $d$-th order orbit information is preserved under projection whenever $d \le m$, it does not address the complementary regime $d>m$. In that case, generic $d$-tuples of frequencies no longer lie in a common $m$-dimensional slice, so the geometric argument underlying the proof breaks down. Consequently, the projected $d$-th order moment need not determine the full $d$-th order Haar-invariant moment of the object. This limitation should not be confused, however, with non-identifiability of the object itself. For example, in the cryo-EM setting one has $n=3$ and $m=2$. For the third moment, $d=3>m$, and hence the full third-order Haar-invariant moment of the unprojected volume is not generally determined by the third-order moment of the projected observations through the slice mechanism considered here. Nevertheless, prior works have shown that third-order projection statistics can still suffice to identify the volume, at least locally or generically, in suitable cryo-EM models~\cite{bandeira2023estimation,fan2024maximum}. Thus, full recovery of a particular $d$-th order invariant moment is a sufficient mechanism for transferring moment-based identifiability results from the unprojected model to the tomographic model, but it is not necessarily a necessary condition for recovering the orbit of $f$. In the regime $d>m$, the relevant question is therefore which components of the higher-order invariant information remain visible from projection data, and whether these accessible components suffice for object recovery.

\section*{Acknowledgment}
T.B. and D.E. are supported in part by BSF under Grant 2020159. T.B. is also supported in part by NSF-BSF under Grant 2024791 and in part by ISF under Grant 1924/21.

\bibliographystyle{plain}
%\bibliography{refs} 

\begin{thebibliography}{10}
	
	\bibitem{anden2018structural}
	Joakim And{\'e}n and Amit Singer.
	\newblock Structural variability from noisy tomographic projections.
	\newblock {\em SIAM Journal on Imaging Sciences}, 11(2):1441--1492, 2018.
	
	\bibitem{balanov2026generalized}
	Amnon Balanov, Tamir Bendory, and Dan Edidin.
	\newblock The generalized method of moments is (almost) statistically efficient
	in low-{SNR} gaussian latent-variable models.
	\newblock {\em arXiv preprint arXiv:2605.30095}, 2026.
	
	\bibitem{baldwin2020non}
	Philip~R Baldwin and Dmitry Lyumkis.
	\newblock Non-uniformity of projection distributions attenuates resolution in
	cryo-{EM}.
	\newblock {\em Progress in Biophysics and Molecular Biology}, 150:160--183,
	2020.
	
	\bibitem{bandeira2023estimation}
	Afonso~S Bandeira, Ben Blum-Smith, Joe Kileel, Jonathan Niles-Weed, Amelia
	Perry, and Alexander~S Wein.
	\newblock Estimation under group actions: recovering orbits from invariants.
	\newblock {\em Applied and Computational Harmonic Analysis}, 66:236--319, 2023.
	
	\bibitem{bendory2020single}
	Tamir Bendory, Alberto Bartesaghi, and Amit Singer.
	\newblock Single-particle cryo-electron microscopy: Mathematical theory,
	computational challenges, and opportunities.
	\newblock {\em IEEE Signal Processing Magazine}, 37(2):58--76, 2020.
	
	\bibitem{bendory2017bispectrum}
	Tamir Bendory, Nicolas Boumal, Chao Ma, Zhizhen Zhao, and Amit Singer.
	\newblock Bispectrum inversion with application to multireference alignment.
	\newblock {\em IEEE Transactions on Signal Processing}, 66(4):1037--1050, 2017.
	
	\bibitem{bendory2024sample}
	Tamir Bendory and Dan Edidin.
	\newblock The sample complexity of sparse multireference alignment and
	single-particle cryo-electron microscopy.
	\newblock {\em SIAM Journal on Mathematics of Data Science}, 6(2):254--282,
	2024.
	
	\bibitem{bendory2026provable}
	Tamir Bendory, Dan Edidin, Josh Katz, Shay Kreymer, and Nir Sharon.
	\newblock Provable orbit recovery over {SO}(3) from the non-uniform second
	moment.
	\newblock {\em arXiv preprint arXiv:2602.20590}, 2026.
	
	\bibitem{bhamre2015orthogonal}
	Tejal Bhamre, Teng Zhang, and Amit Singer.
	\newblock Orthogonal matrix retrieval in cryo-electron microscopy.
	\newblock In {\em 2015 IEEE 12th International Symposium on Biomedical Imaging
		(ISBI)}, pages 1048--1052. IEEE, 2015.
	
	\bibitem{fan2024maximum}
	Zhou Fan, Roy~R Lederman, Yi~Sun, Tianhao Wang, and Sheng Xu.
	\newblock Maximum likelihood for high-noise group orbit estimation and
	single-particle cryo-{EM}.
	\newblock {\em Annals of Statistics}, 52(1):52, 2024.
	
	\bibitem{gilles2025cryo}
	Marc~Aurele Gilles and Amit Singer.
	\newblock Cryo-{EM} heterogeneity analysis using regularized covariance
	estimation and kernel regression.
	\newblock {\em Proceedings of the National Academy of Sciences},
	122(9):e2419140122, 2025.
	
	\bibitem{herman2009fundamentals}
	Gabor~T Herman.
	\newblock {\em Fundamentals of computerized tomography: image reconstruction
		from projections}.
	\newblock Springer Science \& Business Media, 2009.
	
	\bibitem{kam1980reconstruction}
	Zvi Kam.
	\newblock The reconstruction of structure from electron micrographs of randomly
	oriented particles.
	\newblock {\em Journal of Theoretical Biology}, 82(1):15--39, 1980.
	
	\bibitem{levin20183d}
	Eitan Levin, Tamir Bendory, Nicolas Boumal, Joe Kileel, and Amit Singer.
	\newblock {3D} ab initio modeling in cryo-{EM} by autocorrelation analysis.
	\newblock In {\em 2018 IEEE 15th International Symposium on Biomedical Imaging
		(ISBI 2018)}, pages 1569--1573. IEEE, 2018.
	
	\bibitem{natterer2001mathematics}
	Frank Natterer.
	\newblock {\em The mathematics of computerized tomography}.
	\newblock SIAM, 2001.
	
	\bibitem{naydenova2017measuring}
	Katerina Naydenova and Christopher~J Russo.
	\newblock Measuring the effects of particle orientation to improve the
	efficiency of electron cryomicroscopy.
	\newblock {\em Nature Communications}, 8(1):629, 2017.
	
	\bibitem{singer2020computational}
	Amit Singer and Fred~J Sigworth.
	\newblock Computational methods for single-particle electron cryomicroscopy.
	\newblock {\em Annual Review of Biomedical Data Science}, 3:163--190, 2020.
	
\end{thebibliography}

\end{document}